\documentclass[submission,copyright,creativecommons]{eptcs}
\providecommand{\event}{ICLP 2019} 
\usepackage{breakurl}              
\usepackage[english]{babel}

\usepackage{mathptmx}
\usepackage{latexsym}
\usepackage{amsfonts,amssymb,latexsym}

\usepackage{amsmath,amsthm}
\usepackage{url}
\usepackage{multirow}
\usepackage{color}
\usepackage{adjustbox}

\usepackage{xypic}
\input xy
\xyoption{all}

\newcommand{\?}{\stackrel{?}{=}}

\def\cE{\mathcal{E}}

\def\cT{\mathcal{T}}

\long\def\comment#1{}

\newcommand{\caE}{\ensuremath{\mathcal E}}

\newcommand{\caT}{\ensuremath{\mathcal T}}

\newcommand{\caX}{\ensuremath{\mathcal X}}

\newcommand{\sort}[1]{\ensuremath{\mathsf{#1}}}

\newcommand{\Variables}{\caX}
\newcommand{\Symbols}{\Sigma}

\newcommand{\TermsOn}[5]{{\caT^{#4}_{#1}(#2)}_{#3}^{#5}}
\newcommand{\Terms}{\TermsOn{\Symbols}{\Variables}{}{}{}}
\newcommand{\TermsS}[1]{\TermsOn{\Symbols}{\Variables}{\sort{#1}}{}{}}

\newcommand{\GTermsOn}[2]{\caT^{#2}_{#1}}
\newcommand{\GTerms}{\GTermsOn{\Symbols}{}}
\newcommand{\GTermsS}[1]{\GTermsOn{\Symbols,\sort{#1}}{}}

\newcommand{\SubstOn}[2]{{\cal S}ubst(#1,#2)}
\newcommand{\Substs}{\SubstOn{\Symbols}{\Variables}{}{}{}}
\newcommand{\idsubst}{\textit{id}}

\newcommand{\composeSubst}{}
\newcommand{\composeRel}{;}
\newcommand{\compose}{\composeSubst}
\newcommand{\restrict}[2]{#1|_{#2}}

\newcommand{\congr}[1]{=_{#1}}

\newcommand{\csuV}[3]{\textit{CSU\/}^{#2}_{#3}({#1})}

\newcommand{\var}[1]{\mathit{Var}(#1)}

\newcommand{\funocc}[1]{\mathit{Pos}_{\Symbols}(#1)}

\newcommand{\subterm}[2]{#1|_{#2}}
\newcommand{\replace}[3]{#1[#3]_{#2}}
\newcommand{\domain}[1]{\mathit{Dom}(#1)}
\newcommand{\range}[1]{\intrvar{#1}}
\newcommand{\intrvar}[1]{\mathit{Ran}(#1)}

\newcommand{\rewrite}[1]{\rightarrow_{#1}}
\newcommand{\rewrites}[1]{\rightarrow^*_{#1}}

\newcommand{\narrow}[2]{\mathop{\stackrel{#1}{\rightsquigarrow}_{#2}}}

\newcommand{\norm}[1]{{\downarrow_ {#1}}}

\newcommand{\pr}[1]{\mbox{\tt #1}}   

\newcommand{\sem}[1]{{[\![#1]\!]}_{E,B}}

\usepackage{ifthen}

\newenvironment{flemma-noname}[2][]{\vskip\topsep\noindent{\bf
Lemma #2\ifthenelse{\equal{#1}{}}{}{\ }#1.}\em\ }{\vskip\topsep}

\newcommand{\ignore}[1]{}

\newtheorem{example}{Example}
\newtheorem{definition}{Definition}
\newtheorem{lemma}[definition]{Lemma}
\newtheorem{proposition}[definition]{Proposition}
\newtheorem{corollary}[definition]{Corollary}

\title{Most General Variant Unifiers\thanks{
This work has been partially supported by the EU (FEDER) and the Spanish MCIU under grant RTI2018-094403-B-C32, by the Spanish Generalitat Valenciana under grants PROMETEO/2019/098 and APOSTD/2019/127, and by the US Air Force Office of Scientific Research under award number FA9550-17-1-0286.
}}
\author{
Santiago Escobar \qquad \qquad Julia Sapi\~{n}a
\institute{VRAIN (Valencian Research Institute for Artificial Intelligence)\\Universitat Polit\`ecnica de Val\`encia\\
Valencia, Spain}
\email{\{sescobar,jsapina\}@upv.es}
}

\def\titlerunning{Most General Variant Unifiers}
\def\authorrunning{S. Escobar \& J. Sapi\~{n}a}
\begin{document}
\maketitle

\begin{abstract}
Equational unification of two terms consists of finding a substitution that, when applied to both terms, makes them equal modulo some equational properties. Equational unification is of special relevance to automated deduction, theorem proving, protocol analysis, partial evaluation, model checking, etc. Several algorithms have been developed in the literature for specific equational theories, such as associative-commutative symbols, exclusive-or, Diffie-Hellman, or Abelian Groups. Narrowing was proved to be complete for unification and several cases have been studied where narrowing provides a decidable unification algorithm. A new narrowing-based equational unification algorithm relying on the concept of the \emph{variants} of a term has been developed and it is available in the most recent version of Maude, version 2.7.1, which provides quite sophisticated unification features. A variant of a term t is a pair consisting of a substitution $\sigma$ and the canonical form of $t\sigma$. Variant-based unification is decidable when the equational theory satisfies the \emph{finite variant property}. However, it may compute many more unifiers than the necessary and, in this paper, we explore how to strengthen the variant-based unification algorithm implemented in Maude to produce a minimal set of most general variant unifiers. Our experiments suggest that this new adaptation of the variant-based unification is more efficient both in execution time and in the number of computed variant unifiers than the original algorithm available in Maude.
\end{abstract}

\section{Introduction}\label{sec:intro}

Equational unification of two terms is of special relevance to many areas in computer science and consists of finding a substitution that, when applied to both terms, makes them equal modulo some equational properties. Several algorithms have been developed in the literature for specific equational theories, such as associative-commutative symbols, exclusive-or, Diffie-Hellman, or Abelian Groups (see~\cite{BS-handbook00}). Narrowing was proved to be complete for unification \cite{Hullot80,JKK83} and several cases have been studied where narrowing provides a decidable unification algorithm~\cite{AEI09,AEI11}. A new narrowing-based equational unification algorithm relying on the concept of the \emph{variants} of a term \cite{CD05} has been developed \cite{ESM12} and it is available in the most recent version of Maude, version 2.7.1, which provides quite sophisticated unification features \cite{maude-manual,Meseguer18-wollic}.

Several tools and techniques rely on Maude's advanced unification capabilities, such as termination~\cite{DLM09} and local confluence and coherence~\cite{DM12} proofs, narrowing-based theorem proving~\cite{Rusu10} or testing~\cite{Riesco14}, and \emph{logical model checking} \cite{EM07,BEM13}. The area of cryptographic protocol analysis has also benefited: the Maude-NPA tool~\cite{EMM09} is the most successful example of using variant-based equational unification in Maude and the Tamarin tool~\cite{MSCB13,DDKS17,DHRS18} also relies on variants. Numerous decision procedures for formula satisfiability modulo equational theories also rely on unification, either based on narrowing \cite{TGRK15} or by using variant generation in finite variant theories~\cite{Meseguer18-scp}.

However, variant-based unification may compute many more unifiers than the necessary. In this paper, we explore how to improve the variant-based unification algorithm implemented in Maude to produce a smaller, yet complete, set of most general variant unifiers. After some preliminaries in Section~\ref{sec:preliminaries}, we recall variant-based unification in Section~\ref{sec:eq_unification} and propose how to compute a set of most general variant unifiers in Section~\ref{sec:mgvu}. In Section~\ref{sec:mgvu-fast}, we propose a new fast algorithm that considerably reduces the number of variant unifiers by computing a complete (yet not always minimal) set of most general unifers modulo the considered theory. Our experiments in Section~\ref{sec:exp} demonstrate that this new adaptation of the variant-based unification is more efficient both in execution time and in the number of computed variant unifiers than the original algorithm. We conclude in Section~\ref{sec:conc}.
\section{Preliminaries}\label{sec:preliminaries}

We follow the classical notation and terminology from \cite{Terese03} for term rewriting, from \cite{BS-handbook00} for unification, and from \cite{Meseguer92} for rewriting logic and order-sorted notions.

We assume an order-sorted signature $\mathsf{\Sigma} = (S, \leq, \Sigma)$ with a poset of sorts $(S, \leq)$. The poset $(\sort{S},\leq)$ of sorts for $\Symbols$ is partitioned into equivalence classes, called \emph{connected components}, by the equivalence relation $(\leq \cup \geq)^+$. We assume that each connected component $[\sort{s}]$ has a \emph{top element} under $\leq$, denoted $\top_{[\sort{s}]}$ and called the \emph{top sort} of $[\sort{s}]$. This involves no real loss of generality, since if $[\sort{s}]$ lacks a top sort, it can be easily added. We also assume an $\sort{S}$-sorted family $\Variables=\{\Variables_\sort{s}\}_{\sort{s} \in \sort{S}}$ of disjoint variable sets with each $\Variables_\sort{s}$ countably infinite. $\TermsS{s}$ is the set of terms of sort \sort{s}, and $\GTermsS{s}$ is the set of ground terms of sort \sort{s}. We write $\Terms$ and $\GTerms$ for the corresponding order-sorted term algebras. Given a term $t$, $\var{t}$ denotes the set of variables in $t$.

A \textit{substitution} $\sigma\in\Substs$ is a sorted mapping from a finite subset of $\Variables$ to $\Terms$. Substitutions are written as $\sigma=\{X_1 \mapsto t_1,\ldots,X_n \mapsto t_n\}$ where the domain of $\sigma$ is $\domain{\sigma}=\{X_1,\ldots,X_n\}$ and the set of variables introduced by terms $t_1,\ldots,t_n$ is written $\range{\sigma}$. The identity substitution is $\idsubst$. Substitutions are homomorphically extended to $\Terms$. The application of a substitution $\sigma$ to a term $t$ is denoted by $t\sigma$ or $\sigma(t)$. For simplicity, we assume that every substitution is idempotent, i.e., $\sigma$ satisfies $\domain{\sigma}\cap\range{\sigma}=\emptyset$. The restriction of $\sigma$ to a set of variables $V$ is $\subterm{\sigma}{V}$, i.e., $\forall x\in V$, $\subterm{\sigma}{V}(x)=\sigma(x)$ and $\forall x\not\in V$, $\subterm{\sigma}{V}(x)=x$. Composition of two substitutions $\sigma$ and $\sigma'$ is denoted by $\sigma\compose\sigma'$. Combination of two substitutions $\sigma$ and $\sigma'$ such that $\domain{\sigma}\cap\domain{\sigma'}=\emptyset$ is denoted by $\sigma \cup \sigma'$. We call an idempotent substitution $\sigma$ a variable \emph{renaming} if there is another idempotent substitution $\sigma^{-1}$ such that $(\sigma\sigma^{-1})|_{Dom(\sigma)} = \idsubst$.

A \textit{$\Symbols$-equation} is an unoriented pair $t = t'$, where $t,t' \in \TermsS{s}$ for some sort $\sort{s}\in\sort{S}$. An \emph{equational theory} $(\Symbols,E)$ is a pair with $\Symbols$ an order-sorted signature and $E$ a set of $\Symbols$-equations. Given $\Symbols$ and a set $E$ of $\Symbols$-equations, order-sorted equational logic induces a congruence relation $\congr{E}$ on terms $t,t' \in \Terms$ (see~\cite{Meseguer97}). Throughout this paper we assume that $\GTermsS{s}\neq\emptyset$ for every sort \sort{s}, because this affords a simpler deduction system. An equational theory $(\Symbols,E)$ is \emph{regular} if for each $t = t'$ in $E$, we have $\var{t} = \var{t'}$. An equational theory $(\Symbols,E)$ is \emph{linear} if for each $t = t'$ in $E$, each variable occurs only once in $t$ and in $t'$. An equational theory $(\Symbols,E)$ is \textit{sort-preserving} if for each $t = t'$ in $E$, each sort \sort{s}, and each substitution $\sigma$, we have $t \sigma \in \TermsS{s}$ iff $t' \sigma \in \TermsS{s}$. An equational theory $(\Symbols,E)$ is \emph{defined using top sorts} if for each equation $t = t'$ in $E$, all variables in $\var{t}$ and $\var{t'}$ have a top sort. Given two terms $t$ and $t'$, we say $t$ is more general than $t'$, denoted as $t \sqsupseteq_{E} t'$, if there is a substitution $\eta$ such that $t\eta \congr{E} t'$. Similarly, given two substitutions $\sigma$ and $\rho$, we say $\sigma$ is more general than $\rho$ for a set $W$ of variables, denoted as $\subterm{\sigma}{W} \sqsupseteq_{E} \subterm{\rho}{W}$, if there is a substitution $\eta$ such that $\subterm{(\sigma\compose\eta)}{W} \congr{E} \subterm{\rho}{W}$. The $\sqsupseteq_{E}$ relation induces an equivalence relation $\simeq_{E}$, i.e., $t \simeq_{E} t'$ iff $t \sqsupseteq_{E} t'$ and $t \sqsubseteq_{E} t'$.

An \textit{$E$-unifier} for a $\Symbols$-equation $t = t'$ is a substitution $\sigma$ such that $t\sigma \congr{E} t'\sigma$. For $\var{t}\cup\var{t'} \subseteq W$, a set of substitutions $\csuV{t = t'}{W}{E}$ is said to be a \textit{complete} set of unifiers for the equality $t = t'$ modulo $E$ away from $W$ iff: (i) each $\sigma \in \csuV{t = t'}{W}{E}$ is an $E$-unifier of $t = t'$; (ii) for any $E$-unifier $\rho$ of $t = t'$ there is a $\sigma \in \csuV{t=t'}{W}{E}$ such that $\subterm{\sigma}{W} \sqsupseteq_{E} \subterm{\rho}{W}$; and (iii) for all $\sigma \in \csuV{t=t'}{W}{E}$, $\domain{\sigma} \subseteq (\var{t}\cup\var{t'})$ and $\range{\sigma} \cap W = \emptyset$. Given a conjunction $\Gamma$ of equations, a set $U$ of $E$-unifiers of $\Gamma$ is said to be \textit{minimal} if it is complete and for all distinct elements $\sigma$ and $\sigma'$ in $U$, $\sigma \sqsupseteq_E \sigma'$ implies $\sigma \congr{E} \sigma'$. A unification algorithm is said to be \textit{finitary} and complete if it always terminates after generating a finite and complete set of unifiers. A unification algorithm is said to be \textit{minimal} and complete if it always returns a minimal and complete set of unifiers.

A \textit{rewrite rule} is an oriented pair $l \to r$, where $l \not\in \Variables$ and $l,r \in \TermsS{s}$ for some sort $\sort{s}\in\sort{S}$. An \textit{(unconditional) order-sorted rewrite theory} is a triple $(\Symbols,E,R)$ with $\Symbols$ an order-sorted signature, $E$ a set of $\Symbols$-equations, and $R$ a set of rewrite rules. The set $R$ of rules is \textit{sort-decreasing} if for each $t \rightarrow t'$ in $R$, each $\sort{s} \in \sort{S}$, and each substitution $\sigma$, $t'\sigma \in \TermsS{s}$ implies $t\sigma \in \TermsS{s}$. The rewriting relation on $\Terms$, written $t \rewrite{p,R} t'$ (or just $t \rewrite{R} t'$) holds between $t$ and $t'$ iff there exist $p \in \funocc{t}$, $l \to r\in R$ and a substitution $\sigma$, such that $\subterm{t}{p} = l\sigma$, and $t' = \replace{t}{p}{r\sigma}$. The relation $\rewrite{R/E}$ on $\Terms$ is ${\congr{E} \composeRel\rewrite{R}\composeRel\congr{E}}$. The transitive (resp. transitive and reflexive) closure of $\rewrite{R/E}$ is denoted $\rewrite{R/E}^+$ (resp. $\rewrites{R/E}$). 

Reducibility of $\rewrite{R/E}$ is undecidable in general since $E$-congruence classes can be arbitrarily large. Therefore, $R/E$-rewriting is usually implemented by $R,E$-rewriting under some conditions on $R$ and $E$ such as confluence, termination, and coherence (see~\cite{JK86,Meseguer17}). A relation $\rewrite{R,E}$ on $\Terms$ is defined as: $t \rewrite{p,R,E} t'$ (or just $t \rewrite{R,E} t'$) iff there is a non-variable position $p \in \funocc{t}$, a rule $l \to r$ in $R$, and a substitution $\sigma$ such that $\subterm{t}{p} \congr{E} l\sigma$ and $t' = \replace{t}{p}{r\sigma}$. The narrowing relation $\narrow{}{R,E}$ on $\Terms$ is defined as: $t \narrow{\sigma}{p,R,E} t'$ (or just $t \narrow{\sigma}{R,E} t'$) iff there is a non-variable position $p \in \funocc{t}$, a rule $l \to r$ in $R$, and a substitution $\sigma$ such that $\subterm{t}{p}\sigma \congr{E} l\sigma$ and $t' = (\replace{t}{p}{r})\sigma$.

We call $(\Symbols,B,E)$ a \emph{decomposition} of an order-sorted equational theory ${(\Symbols,E\uplus B)}$ if $B$ is regular, linear, sort-preserving, defined using top sorts, and has a finitary and complete unification algorithm, which implies that $B$-matching is decidable, and equations $E$ are oriented into rules $\overrightarrow{E}$ such that they are sort-decreasing and \emph{convergent}, i.e., confluent, terminating, and strictly coherent modulo $B$  \cite{DM12,LM16,Meseguer17}. The irreducible version of a term $t$ is denoted by $t\norm{R,E}$.

Given a decomposition $(\Symbols,B,E)$ of an equational theory and a term $t$, a pair $(t',\theta)$ of a term $t'$ and a substitution $\theta$ is an $E,B$-\emph{variant} (or just a variant) of $t$ if $t\theta\norm{E,B} \congr{B} t'$ and $\theta\norm{E,B} \congr{B} \theta$ ~\cite{CD05,ESM12}. A \emph{complete set of $E,B$-variants}~\cite{ESM12} (up to renaming) of a term $t$ is a subset, denoted by $\sem{t}$, of the set of all $E,B$-variants of $t$ such that, for each $E,B$-variant $(t',\sigma)$ of $t$, there is an $E,B$-variant $(t'', \theta) \in \sem{t}$ such that $(t'',\theta) \sqsupseteq_{E,B} (t',\sigma)$, i.e., there is a substitution $\rho$ such that $t' \congr{E} t''\rho$ and $\restrict{\sigma}{\var{t}} =_{E} \restrict{(\theta\rho)}{\var{t}}$. A decomposition $(\Symbols,B,E)$ has the \emph{finite variant property} (FVP)~\cite{ESM12} (also called a \emph{finite variant decomposition}) iff for each $\Symbols$-term $t$, there exists a complete and finite set $\sem{t}$ of variants of $t$. Note that whether a decomposition has the finite variant property is undecidable~\cite{BGLN13} but a technique based on the dependency pair framework has been developed in \cite{ESM12} and a semi-decision procedure that works well in practice is available in~\cite{CME14tr}.

\section{Variant-based Equational Unification in Maude 2.7.1}\label{sec:eq_unification}

Rewriting logic \cite{Meseguer92} is  a flexible semantic framework within which  different concurrent systems can be naturally specified 
(see \cite{Meseguer12}). Rewriting Logic is efficiently implemented in the high-performance system Maude~\cite{maude-manual}, which has itself a formal environment of verification tools thanks to its reflective capabilities (see \cite{Maude07,Meseguer12}).

Since 2007, several symbolic capabilities have been successively added to Maude (see~\cite{DEEM+18,Meseguer18-wollic} and references therein). First, Maude has been endowed with unification, i.e., \emph{order-sorted equational unification}. Second, Maude has been extended with symbolic reachability features that rely on Maude's unification, i.e., \emph{narrowing-based reachability analysis} as well as the more general \emph{symbolic LTL model checking of infinite-state systems}~\cite{EM07,BEM13}. However, Maude's unification features are quite general in nature: (i) they are applicable to order-sorted signatures; (ii) they work modulo any combination of the equational axioms of associativity (A), commutativity (C), and identity (U); and (iii) they work modulo a set of equations that are assumed convergent modulo axioms. The third part is supported via the concept of the \emph{variants} of a term \cite{CD05} and the \emph{folding variant narrowing strategy} \cite{ESM12}, which achieves termination when the equational theory has the \emph{finite variant property} \cite{CD05,ESM12}. All these unification capabilities are seamlessly provided by a variant-based unification command in Maude, as shown below.

Equational unification can be simply understood as variant computation in an extended equational theory.

\begin{definition}\label{def:extended}{\rm\cite{ESM12}}
Given a decomposition $(\Symbols,B,E)$ with a poset of sorts $(\sort{S}, \leq)$ of an equational theory $(\Symbols,\caE)$, we extend $(\Symbols,B,E)$ and $(\sort{S}, \leq)$ to $(\widehat{\Symbols},B,\widehat{E})$ and $(\sort{\widehat{S}}, \leq)$ as follows:
\begin{enumerate}
\item we add a new sort \sort{Truth} to $\sort{\widehat{S}}$, not related to any sort in $\Sigma$,
\item we add a constant operator $\pr{tt}$ of sort $\sort{Truth}$ to $\widehat{\Symbols}$,
\item for each top sort of a connected component \sort{[s]}, we add an operator $\pr{eq}$ : \sort{[s]} $\times$ \sort{[s]} $\rightarrow$ \sort{Truth} to $\widehat{\Symbols}$, and
\item for each top sort $\sort{[s]}$, we add a variable $X{:}\sort{[s]}$ and an extra rule $\pr{eq}(X{:}\sort{[s]},X{:}\sort{[s]}) \rightarrow \pr{tt}$ to $\widehat{E}$.
\end{enumerate}
\end{definition}
Then, given any two $\Symbols$-terms $t,t'$, if $\theta$ is an equational unifier of $t$ and $t'$, then the $E{,}B$-canonical forms of $t\theta$ and $t'\theta$ must be $B$-equal and therefore the pair $(\pr{tt},\theta)$ must be a variant of the term $\pr{eq}(t,t')$.  Furthermore, if the term $\pr{eq}(t,t')$ has a finite set of most general variants, then we are \emph{guaranteed} that the set of most general $\mathcal{E}$-unifiers of $t$ and $t'$ is \emph{finite}.

Let us make explicit the relation between variants and equational unification. First, we define the intersection of two sets of variants. Without loss of generality, we assume in this paper that each variant pair $(t',\sigma)$ of a term $t$ uses new freshly generated variables.

\begin{definition}[Variant Intersection]{\rm\cite{ESM12}}
Given a decomposition $(\Symbols,B,E)$ of an equational theory, two $\Symbols$-terms $t_1$ and $t_2$ such that $W_\cap = \var{t_1}\cap\var{t_2}$ and $W_\cup = \var{t_1}\cup\var{t_2}$, and two sets $V_1$ and $V_2$ of variants of $t_1$ and $t_2$, respectively, we define 
$V_1 \cap V_2 = 
\{(u_1\sigma,\theta_1\sigma \cup \theta_2\sigma \cup \sigma) \mid (u_1,\theta_1) \in V_1 \wedge (u_2,\theta_2) \in V_2 \wedge 
\exists \sigma: \sigma \in \csuV{u_1 = u_2}{W_\cup}{B} 
\wedge 
\restrict{(\theta_1\sigma)}{W_\cap} 
\congr{B} 
\restrict{(\theta_2\sigma)}{W_\cap}
\}
$.
\end{definition}

Then, we define variant-based unification as the computation of the variants of the two terms in a unification problem and their intersection.

\begin{proposition}[Variant-based Unification]{\rm\cite{ESM12}}
Let $(\Symbols,B,E)$ be a decomposition of an equational theory. Let $t_{1},t_{2}$ be two $\Symbols$-terms. Then, $\rho$ is an unifier of $t_{1}$ and $t_{2}$ iff $\exists (t',\rho)\in\sem{t_{1}}\cap\sem{t_{2}}$.
\end{proposition}

The most recent version 2.7.1 of Maude \cite{maude-manual} incorporates variant-based unification based on the folding variant narrowing strategy \cite{ESM12}. First, there exists a variant generation command of the form: 

\noindent

{\small
\begin{verbatim}
  get variants [ n ] in ModId : Term .
\end{verbatim}
}

\noindent
where $n$ is an optional argument providing a bound on the number of variants requested, so that if the cardinality of the set of variants is greater than the specified bound, the variants beyond that bound are omitted; and \texttt{ModId} is the identifier of the module where the command takes place. Second, there exists a variant-based unification command of the form: 

\noindent

{\small
\begin{verbatim}
  variant unify [ n ] in ModId : T1 =? T1' /\ ... /\  Tk =? Tk' .
\end{verbatim}
}

\noindent
where $k\geq 1$ and $n$ is an optional argument providing a bound on the number of unifiers requested, so that if there are more unifiers, those beyond that bound are omitted; and \texttt{ModId} is the identifier of the module where the command takes place.

\begin{example}\label{ex:xor}
Consider the following equational theory for exclusive-or that assumes three extra constants \verb+a+, \verb+b+, and \verb+c+. Note that the theory is not coherent modulo $AC$ without the second equation.

{\small
\begin{verbatim}
  fmod EXCLUSIVE-OR is 
    sorts Elem ElemXor .  
    subsort Elem < ElemXor .
    ops a b c : -> Elem .
    op mt : -> ElemXor .
    op _*_ : ElemXor ElemXor -> ElemXor [assoc comm] .
    vars X Y Z U V : [ElemXor] .
    eq [idem] :     X * X = mt    [variant] .
    eq [idem-Coh] : X * X * Z = Z [variant] .
    eq [id] :       X * mt = X    [variant] .
  endfm
\end{verbatim}
}

\noindent
The attribute \verb+variant+ specifies that these equations will be used for variant-based unification. Since this theory has the finite variant property (see \cite{CD05,ESM12}), given the term \verb!X * Y! it is easy to verify that there are seven most general variants.

{\small
\begin{verbatim}
  Maude> get variants in EXCLUSIVE-OR : X * Y .

  Variant #1                                      ...        Variant #7
  [ElemXor]: #1:[ElemXor] * #2:[ElemXor]          ...        [ElemXor]: %1:[ElemXor]
  X --> #1:[ElemXor]                              ...        X --> %1:[ElemXor]
  Y --> #2:[ElemXor]                              ...        Y --> mt
\end{verbatim}
}

\noindent  Note that there are two forms of fresh variables, {\small\textit{\texttt{\#n:Sort}}} and {\small\textit{\texttt{\%n:Sort}}}, depending on whether they are generated by unification modulo axioms or by narrowing with the equations modulo axioms. Also note that the two forms have different counters. 

When we consider a variant unification problem between terms $X * Y$ and $U * V$, there are $57$ unifiers:

{\small
\begin{verbatim}
  Maude> variant unify in EXCLUSIVE-OR : X * Y =? U * V  .
  Unifier #1
  X --> %1:[ElemXor] * %3:[ElemXor]
  Y --> %2:[ElemXor] * %4:[ElemXor]
  V --> %1:[ElemXor] * %2:[ElemXor]
  U --> %3:[ElemXor] * %4:[ElemXor]

  Unifier #2
  X --> %1:[ElemXor] * %3:[ElemXor]
  Y --> %2:[ElemXor]
  V --> %1:[ElemXor] * %2:[ElemXor]
  U --> %3:[ElemXor]
  ...
\end{verbatim}
}

\end{example}

Note that this method does not provide an equational unification algorithm in general: given an equational theory $(\Sigma,\cE)$ and two terms $t,t'$ that have a finite, minimal, and complete set of equational unifiers modulo $\caE$, the equational theory $\caE$ may not have a finite variant decomposition. An example is the unification under homomorphism (or one-side distributivity), where there is a finite number of unifiers of two terms but the theory does not satisfy the finite variant property (see \cite{CD05,ESM12}). 

The following result from \cite{ESM12} ensures a complete set of unifiers for a finite variant decomposition.

\begin{corollary}[Finitary $\caE$-unification]{\rm\cite{ESM12}}
\label{cor:finitary-unification}
Let $(\Symbols,B,E)$ be a finite variant decomposition of an equational theory. Given two terms $t,t'$, the set 
$\csuV{t = t'}{\cap}{E\cup B}=\{\theta \mid (w,\theta)\in\sem{t}\cap\sem{t'}\}$ is a \emph{finite and complete} set of unifiers for $t = t'$.
\end{corollary}

However, Corollary~\ref{cor:finitary-unification} does not provide a minimal set of \emph{most general} unifiers w.r.t. the $\sqsupseteq_{E\cup B}$ relation. 
For instance, it is well-known that unification in the exclusive-or theory is unitary, i.e., there exists only one most general unifier modulo exclusive-or
\cite{KN87}.
For the unification problem $X * Y \? U * V$ of Example~\ref{ex:xor}, the  most general unifier w.r.t. $\sqsupseteq_{E\cup B}$ is 
$\{X \mapsto Y * U * V\}$,
which should be appropriately written as 
$$\sigma=\{X \mapsto Y' * U' * V', Y \mapsto Y', U \mapsto U', V \mapsto V'\}.$$
Note that 
$\{Y \mapsto X * U * V\}$,
$\{U \mapsto Y * X * V\}$,
and
$\{V \mapsto Y * U * X\}$ are equivalent to the former unifier w.r.t. $\sqsupseteq_{E\cup B}$
by composing $\sigma$ with, respectively,
$\rho_1=\{Y' \mapsto X'' * U'' * V'',X' \mapsto X'', U' \mapsto U'', V' \mapsto V''\}$,
$\rho_2=\{U' \mapsto Y'' * X'' * V'',X' \mapsto X'', Y' \mapsto Y'', V' \mapsto V''\}$,
and
$\rho_3=\{V' \mapsto Y'' * U'' * X'',X' \mapsto X'', U' \mapsto U'', Y' \mapsto Y''\}$.
Similarly,
$\{X \mapsto U, Y \mapsto V\}$ 
and
$\{X \mapsto V, Y \mapsto U\}$ 
are equivalent to all the previous ones.

\section{Computing More General Variant Unifiers}\label{sec:mgvu}

Note that when  $(\Sigma,B,E)$ is a finite variant decomposition and $B$-unification is finitary, we get an \emph{$E\cup B$-matching algorithm} as $\textit{Match}_{E\cup B}(u,v)=\{\theta \mid \bar{\theta} \in \csuV{u = \bar{v}}{\cap}{E\cup B}\}$, where $\bar{v}$ is obtained from $v$ by turning its variables $x_1,\ldots,x_n$ into fresh constants $\bar{x}_1,\ldots,\bar{x}_n$, and $\theta$ is obtained from $\bar{\theta}$ by, given a binding $x \mapsto \bar{t}\in\bar{\theta}$, adding the binding $x\mapsto t$ to $\theta$; the term $t$ is easily obtained from $\bar{t}$ by replacing every occurrence of a fresh constant $\bar{x}_1,\ldots,\bar{x}_n$ by its original. We say $t \sqsupseteq_{E\cup B} t'$ if $\textit{Match}_{E\cup B}(t,t')\neq\emptyset$, and $t \sqsupset_{E\cup B} t'$ if $t \sqsupseteq_{E\cup B} t'$ and $t \not=_{E\cup B} t'$.

It is easy to provide, at the theoretical level, a minimal set of most general variant unifiers by post-filtering the set of computed unifiers by using $\sqsupseteq_{E\cup B}$.

\begin{proposition}[Post-filtered Variant-based Unification]\label{prop:post}
Let $(\Symbols,B,E)$ be a finite variant decomposition of an equational theory. Given two terms $t,t'$, the set $\csuV{t = t'}{\cap,\sqsupset}{E\cup B}= \{\theta \mid \theta \in \csuV{t = t'}{\cap}{E\cup B} \wedge \nexists \theta'\in \csuV{t = t'}{\cap}{E\cup B}\setminus\{\theta\} : \theta' \sqsupset_{E\cup B} \theta\}$ is a \emph{finite and complete} set of unifiers for $t = t'$. 
Even more, the quotient $\csuV{t = t'}{\cap,\sqsupset}{E\cup B}/_{\simeq_{E\cup B}}$ w.r.t.\ the equivalence relation $\simeq_{E\cup B}$ induced from $\sqsupseteq_{E\cup B}$ is a \emph{finite, minimal, and complete} set of unifiers for $t = t'$.
\end{proposition}

We have implemented both post-filtering stages  $\csuV{t = t'}{\cap,\sqsupset}{E\cup B}$ and $\csuV{t = t'}{\cap,\sqsupset}{E\cup B}/_{\simeq_{E\cup B}}$ in an extended version of Full Maude version 27g \cite{full-maude} available at \url{http://safe-tools.dsic.upv.es/mgvu}. The new command implementing the algorithm $\csuV{t = t'}{\cap,\sqsupset}{E\cup B}$ is as follows:

\noindent

{\small
\begin{verbatim}
 (post variant unify [ n ] in ModId : T1 =? T1' /\ ... /\  Tk =? Tk' .)
\end{verbatim}
}

\noindent where $k\geq 1$ and $n$ is an optional argument providing a bound on the number of unifiers requested, so that if there are more unifiers, those beyond that bound are omitted; and \texttt{ModId} is the identifier of the module where the command takes place.

When we consider the previous variant unification problem between terms $X * Y$ and $U * V$, now we get just $7$ unifiers from the $57$ unifiers above.

{\small
\begin{verbatim}
  Maude> (post variant unify in EXCLUSIVE-OR : X * Y =? U * V  .)
  Unifier #1                                              ...     Unifier #7
  X --> %1:[ElemXor] * %3:[ElemXor]                       ...     X --> %2:[ElemXor]
  Y --> %2:[ElemXor] * %4:[ElemXor]                       ...     Y --> %1:[ElemXor]
  V --> %1:[ElemXor] * %2:[ElemXor]                       ...     V --> %1:[ElemXor]
  U --> %3:[ElemXor] * %4:[ElemXor]                       ...     U --> %2:[ElemXor]
\end{verbatim}
}

\noindent The new command reporting the quotient $\csuV{t = t'}{\cap,\sqsupset}{E\cup B}/_{\simeq_{E\cup B}}$ is as follows:

\noindent

{\small
\begin{verbatim}
 (post quotient variant unify [ n ] in ModId : T1 =? T1' /\ ... /\  Tk =? Tk' .)
\end{verbatim}
}

\noindent where $k\geq 1$ and $n$ is an optional argument providing a bound on the number of unifiers requested, so that if there are more unifiers, those beyond that bound are omitted; and \texttt{ModId} is the identifier of the module where the command takes place.

When we consider the previous variant unification problem between terms $X * Y$ and $U * V$, now we get just one unifier, since all the seven unifiers reported before are equivalent modulo exclusive-or.

{\small
\begin{verbatim}
  Maude> (post quotient variant unify in EXCLUSIVE-OR : X * Y =? U * V  .)
  
  Unifier #1
  X --> %1:[ElemXor] * %3:[ElemXor]
  Y --> %2:[ElemXor] * %4:[ElemXor]
  V --> %1:[ElemXor] * %2:[ElemXor]
  U --> %3:[ElemXor] * %4:[ElemXor]
\end{verbatim}
}

\section{Fast Computation of More General Variant Unifiers}\label{sec:mgvu-fast}

The computation of both $\csuV{t = t'}{\cap,\sqsupset}{E\cup B}$ and $\csuV{t = t'}{\cap,\sqsupset}{E\cup B}/_{\simeq_{E\cup B}}$ is extremely expensive (see Section~\ref{sec:exp} below), both in execution time and memory usage, because we must use the same variant-based unification command in Maude for obtaining the variant unifiers and then for filtering them. In this section, we provide the main contribution of this paper on improving the computation of a set of most general variant unifiers. Let us motivate our main results with an example.

When we consider a variant unification problem between terms $X$ and $U * V$, we get an explosion of all the variants of $U * V$.

{\small
\begin{verbatim}
  Maude> variant unify in EXCLUSIVE-OR : X =? U * V  .

  Unifier #1
  X --> %1:[ElemXor] * %2:[ElemXor]
  V --> %1:[ElemXor]
  U --> %2:[ElemXor]

  Unifier #2
  X --> mt
  V --> #1:[ElemXor]
  U --> #1:[ElemXor]

  Unifier #3
  X --> #2:[ElemXor] * #3:[ElemXor]
  V --> #1:[ElemXor] * #2:[ElemXor]
  U --> #1:[ElemXor] * #3:[ElemXor]

  Unifier #4
  X --> #1:[ElemXor]
  V --> #1:[ElemXor] * #2:[ElemXor]
  U --> #2:[ElemXor]

  Unifier #5
  X --> #1:[ElemXor]
  V --> #2:[ElemXor]
  U --> #1:[ElemXor] * #2:[ElemXor]

  Unifier #6
  X --> #1:[ElemXor]
  V --> mt
  U --> #1:[ElemXor]

  Unifier #7
  X --> #1:[ElemXor]
  V --> #1:[ElemXor]
  U --> mt
\end{verbatim}
}

\noindent but it is clear that the simplest, most general unifier is $\{X \mapsto U * V\}$

{\small
\begin{verbatim}
  Maude> (post quotient variant unify in EXCLUSIVE-OR : X =? U * V  .)

  Unifier #1
  X --> %1:[ElemXor] * %2:[ElemXor]
  V --> %1:[ElemXor]
  U --> %2:[ElemXor]
\end{verbatim}
}

The main idea here, common to any unification algorithm (see \cite{BS-handbook00}), is that when a variable is found, i.e., $X \? t$, there is no need to search for further unifiers, since any other unifier will be an instance of $X \mapsto t$. We have formalized this idea but extended it to the case of having any context $C[X] \? C[t]$. Indeed, we have formalized it for the very general case of having any context modulo $B$, i.e., $C_1[X] \? C_2[t]$ s.t. $C_1[\Box] =_{B} C_2[\Box]$. The following auxiliary result stating that it is possible that any narrowing step from $t$ does not interfere with $C_1$, $C_2$ and $X$ is essential.

\begin{lemma}\label{lem}
Given a decomposition $(\Symbols,B,E)$ of an equational theory, two $\Symbols$-terms $t_1$ and $t_2$ s.t.
$W_\cap = \var{t_1}\cap\var{t_2}$ and
$W_\cup = \var{t_1}\cup\var{t_2}$,
$(u_1,\theta_1) \in \sem{t_1}$,
$(u_2,\theta_2) \in \sem{t_2}$,
$\sigma \in \csuV{u_1 = u_2}{W_\cup}{B}$ s.t.
$\restrict{(\theta_1\sigma)}{W_\cap} 
\congr{B} 
\restrict{(\theta_2\sigma)}{W_\cap}$,
$(u'_1,\theta'_1) \in \sem{t_1}$ s.t. $(u'_1,\rho) \in \sem{u_1}$ and $\restrict{\theta'_1}{W_\cup} \congr{B} \restrict{\theta_1\rho}{W_\cup}$,
$\sigma' \in \linebreak \csuV{u'_1 = u_2}{W_\cup}{B}$ s.t.
$\restrict{(\theta'_1\sigma)}{W_\cap} 
\congr{B} 
\restrict{(\theta_2\sigma)}{W_\cap}$,
and 
$\domain{\sigma}\cap\domain{\rho}=\emptyset$, 
then
$\restrict{((\theta_1\cup\theta_2)\sigma)}{W_\cup}$ and $\restrict{((\theta'_1\cup\theta_2)\sigma')}{W_\cup}$ are both equational unifiers of $t_1$ and $t_2$
but 
$\restrict{((\theta_1\cup\theta_2)\sigma)}{W_\cup} \sqsupseteq_{E\cup B} \restrict{((\theta'_1\cup\theta_2)\sigma')}{W_\cup}$.
\end{lemma}
\begin{proof}
The statement of the Lemma is depicted in Figure~\ref{fig:lem}.
The proof is done by realizing that
$\domain{\sigma}\cap\domain{\rho}=\emptyset$ 
implies that
$\restrict{(((\theta_1\cup\theta_2)(\sigma\cup\rho))}{W_\cup}$
is also a unifier of $t_1$ and $t_2$, and then
$$\restrict{(\theta_1\cup\theta_2)\sigma}{W_\cup} 
\sqsupseteq_{E\cup B} 
\restrict{(\theta_1\cup\theta_2)(\sigma\cup\rho)}{W_\cup} 
\sqsupseteq_{E\cup B} 
\restrict{(\theta_1\cup\theta_2)\sigma\rho\sigma'}{W_\cup}
=_{B}
\restrict{(\theta'_1\cup\theta_2)\sigma'}{W_\cup}$$
\end{proof}

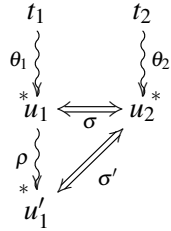
\begin{figure}[t]
$$
\xymatrix{
t_1\ar@{~>}_{\theta_1}_>{*}[d] & t_2\ar@{~>}^{\theta_2}^>{*}[d]\\
u_1\ar@{~>}_{\rho}_>{*}[d] & u_2\ar@{<=>}^{\sigma}[l]\ar@{<=>}^{\sigma'}[ld]\\
u'_1
}
$$
\caption{Sketch of the proof of Lemma~\ref{lem}}\label{fig:lem}
\end{figure}

We redefine the intersection of two sets of variants. Note that this definition does not prevent the generation of the variants of both terms in an unification problem; techniques for avoiding the generation of variants are outside the scope of this paper.

\begin{definition}[Fast Variant Intersection]
Given a decomposition $(\Symbols,B,E)$ of an equational theory, two $\Symbols$-terms $t_1$ and $t_2$ such that $W_\cap = \var{t_1}\cap\var{t_2}$ and $W_\cup = \var{t_1}\cup\var{t_2}$, and two sets $V_1$ and $V_2$ of variants of $t_1$ and $t_2$, respectively, we define 

\noindent
\begin{align*}
V_1 \doublecap V_2 = 
\{
&(u_1\sigma,\theta_1\sigma \cup \theta_2\sigma \cup \sigma) \mid (u_1,\theta_1) \in V_1 \wedge (u_2,\theta_2) \in V_2 \wedge\\ 
&\exists \sigma: \sigma \in \csuV{u_1 = u_2}{W_\cup}{B} 
\wedge 
\restrict{(\theta_1\sigma)}{W_\cap} 
\congr{B} 
\restrict{(\theta_2\sigma)}{W_\cap}
\wedge\\
&
(\nexists (u'_1,\theta'_1) \in V_1,\nexists\rho:
(u_1,\rho)\in\sem{u'_1}\wedge\\
&\ \nexists \sigma': \sigma' \in \csuV{u'_1 = u_2}{W_\cup}{B} 
\wedge 
\restrict{(\theta'_1\sigma')}{W_\cap} 
\congr{B} 
\restrict{(\theta_2\sigma')}{W_\cap}
\wedge 
\domain{\sigma'}\cap\domain{\rho}=\emptyset)
\wedge\\
&
(\nexists (u'_2,\theta'_2) \in V_2,\nexists\rho:
(u_2,\rho)\in\sem{u'_2}\wedge\\
&\ \nexists \sigma': \sigma' \in \csuV{u_1 = u'_2}{W_\cup}{B} 
\wedge 
\restrict{(\theta_1\sigma')}{W_\cap} 
\congr{B} 
\restrict{(\theta'_2\sigma')}{W_\cap}
\wedge 
\domain{\sigma'}\cap\domain{\rho}=\emptyset)
\}
\end{align*}
\end{definition}

Then, we define variant-based unification as the computation of the variants of the two terms in a unification problem and their \emph{minimal} intersection; its proof is immediate by Lemma~\ref{lem}.

\begin{proposition}[Fast Variant-based Unification]\label{prop:mgvu-fast} 
Let $(\Symbols,B,E)$ be a finite variant decomposition of an equational theory. Given two terms $t,t'$, on the one hand, the set $\csuV{t = t'}{\doublecap}{E\cup B}=\{\theta \mid (w,\theta)\in\sem{t}\doublecap\sem{t'}\}$ is a \emph{finite and complete} set of \emph{unifiers} for $t = t'$ and the quotient $\csuV{t = t'}{\doublecap}{E\cup B}/_{\simeq_{E\cup B}}$ is also a (generally smaller) \emph{finite and complete} set of unifiers for $t = t'$. On the other hand, the set $\csuV{t = t'}{\doublecap,\sqsupset}{E\cup B}= \{\theta \mid \theta \in \csuV{t = t'}{\doublecap}{E\cup B} \wedge \nexists \theta'\in \csuV{t = t'}{\doublecap}{E\cup B}\setminus\{\theta\} : \theta' \sqsupset_{E\cup B} \theta\}$ is a \emph{finite and complete} set of unifiers for $t = t'$. Furthermore, the quotient $\csuV{t = t'}{\doublecap,\sqsupset}{E\cup B}/_{\simeq_{E\cup B}}$ is a \emph{finite, minimal, and complete} set of unifiers for $t = t'$.
\end{proposition}

We have implemented these four fast unification methods in an extended version of Full Maude version 27g \cite{full-maude}, which is available  at \url{http://safe-tools.dsic.upv.es/mgvu}:

\begin{itemize}
\item The new command implementing the algorithm $\csuV{t = t'}{\doublecap}{E\cup B}$ is
\noindent
{\small
\begin{verbatim}
(fast variant unify [ n ] in ModId : T1 =? T1' /\ ... /\  Tk =? Tk' .)
\end{verbatim}
}

\item The new command implementing the algorithm $\csuV{t = t'}{\doublecap}{E\cup B}/_{\simeq_{E\cup B}}$ is
\noindent
{\small
\begin{verbatim}
(fast quotient variant unify [ n ] in ModId : T1 =? T1' /\ ... /\  Tk =? Tk' .)
\end{verbatim}
}

\item The new command implementing the algorithm $\csuV{t = t'}{\doublecap,\sqsupset}{E\cup B}$ is
\noindent
{\small
\begin{verbatim}
(fast post variant unify [n] in ModId : T1 =? T1' /\ ... /\  Tk =? Tk' .)
\end{verbatim}
}

\item And the new command implementing the algorithm $\csuV{t = t'}{\doublecap,\sqsupset}{E\cup B}/_{\simeq_{E\cup B}}$ is
\noindent
{\small
\begin{verbatim}
(fast post quotient variant unify [n] in ModId : T1 =? T1'/\ ... /\ Tk =? Tk' .)
\end{verbatim}
}
\end{itemize}

For the unification problem $X * Y$ and $U * V$, the \texttt{fast} command delivers $8$ unifiers instead of the $57$ unifiers for standard variant unification. However, $7$ of those $8$ unifiers are equivalent, thus the \texttt{fast\;quotient} command delivers only $2$ unifiers. Likewise, the \texttt{fast\;post} command returns the same $7$ unifiers as the \texttt{post} command, and the \texttt{fast\;post\;quotient} command gets the same (most general) unifier as the \texttt{post\;quotient} command above. Note that the \texttt{fast} unification command and the \texttt{fast\;quotient} unification command compute these unifiers in a fraction of time compared to the \texttt{post} unification command and the \texttt{post\;quotient} unification command (see unification problem $P_6$ in Section~\ref{sec:exp}). 

When we consider the previous variant unification problem between terms $X$ and $U * V$, now we get just one unifier as desired, and again in a fraction of time compared to $\csuV{t = t'}{\cap,\sqsupset}{E\cup B}$ (see unification problem $P_1$ in Section~\ref{sec:exp}). 

{\small
\begin{verbatim}
  Maude> (fast variant unify in EXCLUSIVE-OR : X =? U * V  .)

  Unifier #1
  X --> %1:[ElemXor] * %2:[ElemXor]
  V --> %1:[ElemXor]
  U --> %2:[ElemXor]
\end{verbatim}
}

\noindent
Note that, in this case, clearly the \texttt{fast\;post} and \texttt{fast\;post\;quotient} unification commands do not improve over the \texttt{fast} unification command.

\section{Experimental Evaluation}\label{sec:exp}

To evaluate the performance of both the post-filtering and the fast unification techniques, we have conducted a series of benchmarks available at \url{http://safe-tools.dsic.upv.es/mgvu}.

All the experiments were conducted on a PC with a 3.3GHz Intel Xeon E5-1660 and 64GB RAM. First, we created a battery of 20 different unification problems for both the {\it exclusive-or} and the {\it abelian group} theories. For each problem and theory, we computed: (i) the unifiers by using the standard {\tt variant\;unify} command provided by the C++ core system of Maude; (ii) the unifiers by using the {\tt post\;quotient\;variant\;unify} command implemented at the metalevel of Maude; (iii) the unifiers by using the {\tt fast\;quotient\;variant\;unify} command implemented at the metalevel of Maude; and (iv) the unifiers by using the {\tt fast\;post\;quotient\;variant\;unify} command, also implemented at the metalevel of Maude. We measured both the number of computed unifiers and the time required for their computation.

Since it is unfair to compare the performance between compiled code and interpreted code, i.e., the C++ core system of Maude and a Maude program using Maude's metalevel, we have reimplemented the {\tt variant unify} command at the metalevel and applied the post-filtering and the fast variant intersection to the output returned by this reimplementation.

Table~\ref{tab:xor} (resp. Table~\ref{tab:ag}) shows the results obtained for the {\it exclusive-or} (resp. {\it abelian group)} theory. \emph{T/O} indicates that a generous $24$ hours \emph{timeout} was reached without any response. The first column describes the unification problem, while the following $\#_{\mathit{maude}}$, $\#_{\mathit{post}}$, $\#_{\mathit{fast}}$, and $\#_{\mathit{fast,post}}$ columns show the number of computed unifiers for Maude's unification command, the post-filtering technique producing the quotient w.r.t. $\sqsupseteq_{E\cup B}$, the fast unification technique, and the combination of fast and the post-filtering, respectively. The $\mathcal{T}_{\mathit{maude}}$ column measures the time (in milliseconds) required to execute the {\tt variant unify} command for the given input problem, the $\cT_{\mathit{post}}$ column measures the time required by the reimplementation of the {\tt variant unify} command together with the post-filtering technique, 
the $\cT_{\mathit{fast}}$ column measures the time required by the reimplementation of the {\tt variant unify} command together with the fast unification technique, and the $\cT_{\mathit{fast,post}}$ column measures the time required of all three combined, the reimplementation, the fast technique, and the post-filtering.

\begin{table}[t]
	\centering
	\scriptsize
	{\setlength{\tabcolsep}{0.5em}
	\begin{tabular*}{\textwidth}{|c|c@{\extracolsep{\fill}}|r|r|r|r|r|r|r|r|}
		\cline{1-10} 
		\multicolumn{2}{|c|}{\it Unification problem}
		&\multicolumn{1}{c|}{$\#_{\mathit{maude}}$}
		&\multicolumn{1}{c|}{$\mathcal{T}_{\mathit{maude}}$} 
		&\multicolumn{1}{c|}{$\#_{\mathit{post}}$}
		&\multicolumn{1}{c|}{$\mathcal{T}_{\mathit{post}}$} 
		&\multicolumn{1}{c|}{$\#_{\mathit{fast}}$}
		&\multicolumn{1}{c|}{$\mathcal{T}_{\mathit{fast}}$} 
		&\multicolumn{1}{c|}{$\#_{\mathit{fast,post}}$}
		&\multicolumn{1}{c|}{$\mathcal{T}_{\mathit{fast,post}}$}\\
		\cline{1-10}
		$P_{1}$ & $ V_1 \? V_2 * V_3$                                   		        &7         &0         &1	&13   		&1		&4        	&1		&4\\
		$P_{2}$ & $ V_1 \? V_2 * V_3 * V_4$                               		        &57        &49        &1    &6545      	&1      &1080      	&1      &1168\\
		$P_{3}$ & $ V_1 \? f_1(V_2 * V_3 * f_1(V_4))$                    		        &21        &3         &1    &199	   	&1   	&47			&1   	&47\\
		$P_{4}$ & $ V_1 \? f_2(V_2 * V_3, f_1(V_2 * V_4))$                		        &61        &98        &1    &18895     	&1      &1463     	&1      &1470\\
		$P_{5}$ & $ V_1 \? f_3(V_2 * V_3,f_1(V_3 * V_4),f_2(V_2,f_1(V_4)))$             &61        &193       &1    &20949     	&1      &1958     	&1      &1966\\
		\cline{1-10}
		$P_{6}$ & $ V_1 * V_2 \? V_3 * V_4$                             		        &57        &10        &1    &12240890  	&2      &72        	&1      &10005912\\
		$P_{7}$ & $ V_1 * V_2 \? f_1(V_3 * V_4)$                                		&28        &8         &1    &697       	&4      &17       	&1      &41\\
		$P_{8}$ & $ V_1 * V_2 \? f_1(V_3 * V_3 * f_1(V_4))$                    			&4         &0         &1    &6         	&4      &3         	&1   	&5\\
		$P_{9}$ & $ V_1 * V_2 \? f_2(V_3 * V_4,f_1(V_3 * V_5))$                			&244       &741       &1    &30490862  	&4      &2193    	&1      &14836\\
		$P_{10}$ & $ V_1 * V_2 \? f_3(V_3 * V_4,f_1(V_4 * V_5),f_2(V_3,f_1(V_5)))$	    &244       &1277      &1    &30423527  	&4      &2868    	&1      &14802\\
		\cline{1-10}
		$P_{11}$ & $ f_1(V_1) \? f_1(V_2 * V_3)$                                	    &7         &0         &1    &13        	&1      &4        	&1      &4\\
		$P_{12}$ & $ f_1(V_1) * f_1(V_2) \? f_1(V_3) * f_1(V_3 * V_4)$                  &13        &3         &2    &118       	&2      &8       	&2      &9\\
		$P_{13}$ & $ f_1(V_1 * V_2) \? f_1(V_3 * V_4 * V_5)$                         	&973       &857       &-    &T/O      	&8      &15539      &-      &T/O\\
		$P_{14}$ & $ f_2(V_1 * V_2,V_2 * V_3) \? f_2(V_4,V_5)$                		    &61        &97        &1    &32836     	&1      &1471     	&1      &1473\\
		$P_{15}$ & $ f_3(V_1 * V_2,V_3 * V_4,V_5 * V_6) \? f_3(V_7,V_8,V_9)$            &343       &173       &1    &165260    	&1      &20608     	&1      &20634\\
		\cline{1-10}
		$P_{16}$ & $ V_1 \? a * b * V_2$                                                &8         &0         &1    &11        	&1      &2         	&1   	&2\\
		$P_{17}$ & $ V_1 * V_2 \? a * b * V_3$                                          &69        &9         &1    &2259      	&5      &74        	&1   	&183\\
		$P_{18}$ & $ V_1 * a \? V_2 * b$                                                &8         &0         &1    &11        	&4      &2         	&1   	&4\\
		$P_{19}$ & $ f_1(a) * f_1(V_1) \? f_1(V_2 * b) * f_1(V_3 * c)$                  &16        &3         &3    &104       	&10     &13        	&3      &47\\
		$P_{20}$ & $ f_2(a,V_1) \? f_2(V_2 * V_3,f_1(a * b))$                           &4         &0         &1    &9         	&4      &4   		&1   	&5\\
		\cline{1-10}
	\end{tabular*}
	}
\caption{Experimental evaluation (exclusive-or)}\label{tab:xor}
\end{table}

\begin{table}[t]
	\centering
	\scriptsize
	{\setlength{\tabcolsep}{0.5em}
	\begin{tabular*}{\textwidth}{|c|c@{\extracolsep{\fill}}|r|r|r|r|r|r|r|r|}
		\cline{1-10} 
		\multicolumn{2}{|c|}{\it Unification problem}
		&\multicolumn{1}{c|}{$\#_{\mathit{maude}}$}
		&\multicolumn{1}{c|}{$\mathcal{T}_{\mathit{maude}}$} 
		&\multicolumn{1}{c|}{$\#_{\mathit{post}}$}
		&\multicolumn{1}{c|}{$\mathcal{T}_{\mathit{post}}$} 
		&\multicolumn{1}{c|}{$\#_{\mathit{fast}}$}
		&\multicolumn{1}{c|}{$\mathcal{T}_{\mathit{fast}}$} 
		&\multicolumn{1}{c|}{$\#_{\mathit{fast,post}}$}
		&\multicolumn{1}{c|}{$\mathcal{T}_{\mathit{fast,post}}$}\\
		\cline{1-10}
		$P_{21}$ & $ V_1 \? V_2 + V_3$                                   		        &47     &68     &1      &6185     &1      &778    &1      &806\\
		$P_{22}$ & $ V_1 \? f_1(V_2 + V_3)$                               		        &47     &68     &1      &6117     &1      &796    &1      &808\\
		$P_{23}$ & $ V_1 \? f_1(V_2 + V_2 + f_1(V_3))$                    		        &8      &13     &1      &125      &1      &43     &1      &43\\
		$P_{24}$ & $ V_1 \? f_2(V_2 + V_3 + f_1(V_3),V_4)$ 		           		        &103    &371    &1      &55662    &1      &10696  &1      &10696\\
		$P_{25}$ & $ V_1 \? f_3(V_2,f_1(V_3 + V_4),f_2(V_3,V_5))$			            &6      &2      &1      &30       &1      &7      &1      &7\\
		\cline{1-10}
		$P_{26}$ & $ V_1 + V_2 \? V_3 + V_4$                             		        &3611   &21663  &-      &T/O      &167    &439304 &-      &T/O\\
		$P_{27}$ & $ V_1 + V_2 \? f_1(V_3 + V_4)$                                		&376    &13864  &1      &22207559 &8      &3830   &1      &27870\\
		$P_{28}$ & $ V_1 + V_2 \? f_1(V_3 + V_3 + f_1(V_4))$                  			&64     &1239   &1      &82170    &8      &904    &1      &3382\\
		$P_{29}$ & $ V_1 + V_2 \? f_2(V_3 + V_4,f_1(V_5))$		               			&376    &13373  &1      &19468887 &8      &4059   &1      &30537\\
		$P_{30}$ & $ V_1 + V_2 \? f_3(V_3 + V_3,V_4,V_5)$							    &32     &466    &1      &4743     &8      &836    &1      &1194\\
		\cline{1-10}
		$P_{31}$ & $ f_1(V_1) \? f_1(V_2 + V_3)$                                	    &47     &71     &1      &9985     &1      &842    &1      &849\\
		$P_{32}$ & $ f_1(V_1) + f_1(V_2) \? f_1(V_3) + f_1(V_3 + V_4)$                  &93     &150    &1      &699872   &1      &1417   &1      &1449\\
		$P_{33}$ & $ f_1(V_1 + V_2) \? f_1(V_3 + - V_4)$                            	&3702   &25277  &-      &T/O      &109    &283851 &1      &48028877\\
		$P_{34}$ & $ f_2(V_1 + V_2,V_2 + V_3) \? f_2(V_4,- V_5)$              		    &188    &356    &1      &154443   &1      &2384   &1      &2409\\
		$P_{35}$ & $ f_3(V_1 + V_2,f_1(V_3), - V_4) \? f_3(V_5,- V_6,V_6)$              &47     &1812   &1      &35992    &1      &25889  &1      &29674\\
		\cline{1-10}
		$P_{36}$ & $ V_1 \? a + - b + V_2$                                              &14     &5      &1      &117      &1      &20     &1      &29\\
		$P_{37}$ & $ V_1 + V_2 \? a + b + V_3$                                          &510    &1411   &1      &1366009  &107    &5557   &1      &288552\\
		$P_{38}$ & $ V_1 + a \? V_2 + b$                                                &14     &9      &1      &107      &8      &8      &1      &63\\
		$P_{39}$ & $ f_1(a) + f_1(V_1) \? f_1(V_2 + - b) + f_1(V_3 + c)$                &12     &17     &2      &277      &2      &150    &2      &142\\
		$P_{40}$ & $ f_2(a,V_1) \? f_2(V_2 + V_3,f_1(a + b))$                           &8      &79     &2      &831      &8      &764    &1      &920\\
		\cline{1-10}
	\end{tabular*}
	}
\caption{Experimental evaluation (abelian group)}\label{tab:ag}
\end{table}

Table~\ref{tab:xor} shows that, for the {\it exclusive-or} theory, the {\tt fast\;post\;quotient} unification command almost replicates the results obtained by using the {\tt post\;quotient} unification command, but in a fraction of time, as in unification problems $P_{9}$ and $P_{10}$. For the number of unifiers, Maude reported $973$ unifiers for the unification problem $P_{13}$, and the fast unification technique delivers just $8$ unifiers, whereas applying the post-filtering technique to either standard or fast unification is hopeless. For the execution time, the unification problem $P_6$ reports only $10$ milliseconds for $\mathcal{T}_{\mathit{maude}}$, $72$ milliseconds for $\mathcal{T}_{\mathit{fast}}$, $12240890$ milliseconds ($3,4$ hours) for $\mathcal{T}_{\mathit{post}}$, and $10005912$ milliseconds ($2,7$ hours) for $\mathcal{T}_{\mathit{fast,post}}$, demonstrating that the post-filtering technique is expensive in any case. 

Table~\ref{tab:ag} shows the experimental results for the {\it abelian group} theory. Since this theory is far more complex than the {\it exclusive-or} theory, the execution time and the number of unifiers are bigger than those in Table~\ref{tab:xor}. For the unification problem $P_{27}$, Maude reported $376$ unifiers and the fast unification technique reported just $8$ unifiers. The post-filtering technique delivers only one most general unifier, but it takes $22207559$ milliseconds ($6,2$ hours) to compute it from the $376$ unifiers and only $27870$ milliseconds (less than $28$ seconds) to compute it from the $8$ unifiers, demonstrating that applying the fast unification technique is advantageous in any case.

\section{Conclusion and Future Work}\label{sec:conc}

The variant-based equational unification algorithm implemented in the most recent version of Maude, version 2.7.1, may compute many more unifiers than the necessary and, in this paper, we have explored how to strengthen such an algorithm to produce a smaller set of variant unifiers. Our experiments suggest that this new adaptation of the variant-based unification is more efficient both in execution time and in the number of computed variant unifiers than the original algorithm.

As far as we know, this is the first work to reduce the number of variant unifiers. The closest work are methods to combine standard unification algorithms with variant-based unification, such as \cite{EKMN+15,EEMR19}. This is just a step forward on developing new techniques for improving variant-based unification and we plan to reduce even more the number of variant unifiers.

\bibliographystyle{eptcs}
\bibliography{biblio}

\end{document}